\documentclass[11pt]{amsart}
\usepackage{color}
\usepackage{amsmath,amsthm,amssymb,mathrsfs, comment}
\usepackage{amsfonts, accents}
\usepackage{times,ulem}
\usepackage{slashed}
\theoremstyle{plain}
\newtheorem{theorem}{Theorem}[section]
\newtheorem{lemma}[theorem]{Lemma}
\newtheorem{lem}[theorem]{Lemma}
\newtheorem{prop}[theorem]{Proposition}

\theoremstyle{definition}

\newtheorem{defn}[theorem]{Definition}
\theoremstyle{remark}
\newtheorem*{remark}{Remark}

\newcommand{\pl}{\partial}
\newcommand{\na}{\nabla}

\newcommand{\lt}{\left}
\newcommand{\rt}{\right}
\newcommand{\rw}{\rightarrow}

\newcommand{\tr}{\mbox{tr}}
\newcommand{\R}{\mathbb{R}}

\title{Quasi-local mass at null infinity in Bondi-Sachs coordinates}
\author{Po-Ning Chen, Mu-Tao Wang, Ye-Kai Wang, and Shing-Tung Yau}

\begin{document}

\begin{abstract}
There are two important statements regarding the Trautman-Bondi mass \cite{BVM, Sachs, Van, Trautman1, Trautman2} at null infinity: one is the positivity \cite{SY, HP}, and the other is the Bondi mass loss formula \cite{BVM}, which are both global in nature. In this note, we compute the limit of the Wang-Yau quasi-local mass on unit spheres at null infinity of an asymptotically flat spacetime in the Bondi-Sachs coordinates. The quasi-local mass leads to a local description of the radiation that is purely gravitational at null infinity. In particular, the quasi-local mass is evaluated in terms of the news function of the Bondi-Sachs coordinates. 

\end{abstract}

\thanks{P.-N. Chen is supported by NSF grant DMS-1308164 and Simons Foundation collaboration grant \#584785, M.-T. Wang is supported by NSF grant DMS-1405152 and DMS-1810856, Y.-K. Wang is supported by MOST Taiwan grant 105-2115-M-006-016-MY2, 107-2115-M-006-001-MY2, and S.-T. Yau is supported by NSF grants  PHY-0714648 and DMS-1308244. The authors would like to thank the National Center for Theoretical Sciences at National Taiwan University where part of this research was carried out}

\maketitle

\section{Introduction}

An observer of the gravitational radiation created by an astronomical event is situated at future null infinity, where light rays emitted from the source approach. The study of the theory of gravitational radiation at null infinity in the last century culminated in a series of papers by Bondi and his collaborators \cite{BVM, Sachs, Van, Trautman1, Trautman2}, in which the Bondi-Trautman mass and the mass loss formula at null infinity are well understood. In particular, the Bondi-Trautman mass was proved to be positive in the work of Schoen-Yau \cite{SY} and Horowitz-Perry \cite{HP}. Both the positivity of mass and the mass loss formula are global statements on null infinity:  knowledge of the mass aspect is required in {\it every} direction. For reasons that are both theoretical and experimental, it is highly desirable to have a quasi-local statement of mass/radiation at null infinity.

In \cite{Chen-Wang-Yau6,Chen-Wang-Yau4}, we embarked on the evaluation the Wang-Yau quasi-local mass on surfaces of fixed size near null infinity of a linear gravitational perturbation of the Schwarzschild spacetime. The ideas and technique in \cite{Chen-Wang-Yau6,Chen-Wang-Yau4} were further developed to address the case of the Vaidya spacetime in \cite{Vaidya}. The construction of these spheres of unit size at null infinity will be reviewed in the next section. In the Vaidya case, we proved in \cite{Vaidya} that the quasi-local mass of a unit size sphere at null infinity is directly related to the derivative of the mass aspect function with respect to the retarded time $u$. In particular, the positivity of the quasi-local mass is implied by the decreasing of the mass aspect function in $u$. In this article, we take on the general case of an asymptotically flat spacetime described in the Bondi-Sachs coordinates. The Vaidya spacetime contains matter which contributes to the radiation. A general vacuum spacetime in the Bondi-Sachs coordinates allows us to investigate radiation that is purely gravitational.

A new ingredient in this article is a variational formula (see Theorem \ref{energy}) which facilitates a much more straightforward computation of the $O(d^{-2})$ than the one in \cite{Vaidya}. Similar to \cite{Vaidya}, it is still crucial to compute the $O(d^{-1})$ term of the optimal embedding. This is done in Lemma 5.1 and Lemma 5.2 of the current article.  As in Lemma 3.3 of \cite{Vaidya}, the optimal embedding equation is reduced to two ordinary differential equations. However, it does not seem possible to obtain explicit solutions to the ODE's as in the Vaidya case. The quasi-local mass is then evaluated by combining Theorem \ref{energy} and the optimal embedding.

The structure of the paper is as follows: in Section 2, we review the general framework of the quasi-local mass at null infinity. In Section 3, we compute the geometric quantities on the spheres at null infinity that are necessary to evaluate the quasi-local mass. In Section 4, we derived the formula for the leading order term of the quasi-local mass. In Section 5, we evaluate the quasi-local mass based on the formula derived in Section 4. See Theorem \ref{main_theorem}. In the last section, Section 6, we look at several special examples.

\section{General framework of quasilocal mass at null infinity}

We consider a null geodesic $\gamma$ parametrized by an affine parameter $d$ with $d_0\leq d<\infty$ and a family of surfaces $\Sigma_d (s)$ for $s>0$ centered at $\gamma(d)$ in the following sense. For each fixed $d$ and $s$, $\Sigma_d (s)$ is a surface that bounds a ball $B_d (s)$ with $\partial B_d(s)=\Sigma_d (s)$,  such that as $s\rightarrow 0$, we have $\lim_{s\rightarrow 0} B_d (s)=\lim_{s\rightarrow 0} \Sigma_d (s)=\gamma(d)$. We evaluate 
the quasilocal mass of $\Sigma_d(s)$ as $d\rightarrow \infty$. In particular, when $s=1$, $\lim_{d\rightarrow \infty} \Sigma_d(1)$ is the unit sphere limit referred on our previous work. 

In practice, such an evaluation is conducted by choosing a family of parametrizations $\mathfrak F_{d}$ from the unit ball $B^3$, $\mathfrak F_{d} : B^3\rightarrow  B_d(1)$ and considering the pull-backs of geometric quantities on $B_d(1)$ as geometric quantities on $B^3$ that depend on the parameter $d$. In particular, $\Sigma_d (s)$ is the image of the sphere of radius $s$ in $B^3$ under $\mathfrak F_{d}$. The unit sphere limit is obtained by setting $s=1$ and taking the limit as $d\rightarrow \infty$. 

When the spacetime is equipped with a global structure at null infinity that corresponds to limits of null geodesics, these unit sphere limits provide information of gravitational radiation observed at null infinity. We illustrate the construction in the Vaidya case where the spacetime metric takes the simple form:
\[- \lt( 1 - \frac{M(u)}{r}\rt) du^2 - 2  dudr  + r^2  d\theta^2 + r^2  \sin^2\theta d\phi^2.\]
 
We first consider a global coordinate change from $(u, r, \theta, \phi)$ to $(t, y^1, y^2, y^3)$ with $t=u+r, y^1= r\sin\theta\sin \phi, y^2= r\sin\theta \cos \phi,$ and $y^3= r\cos \theta$.
In terms of the new coordinate system $(t, y^1, y^2, y^3)$, the parametrization is then given by 
\[\mathfrak F_{d}=(s,\hat{\theta}, \hat{\phi})\rightarrow (t, y^1, y^2, y^3)=(d, d\tilde{d}_1+s\sin\hat{\theta} \sin \hat{\phi}, d \tilde{d}_2+ s\sin\hat{\theta} \cos \hat{\phi}, d\tilde{d}_3+s\cos\hat{\theta}),
\] where $(s,\hat{\theta}, \hat{\phi})$ is a coordinate system on $B^3$ and the constants $(\tilde{d}_1, \tilde{d}_2, \tilde{d_3})$ satisfies $\tilde{d}_1^2 +\tilde{d}_2^2+\tilde{d}_3^2=1$ and indicates the direction of the null geodesic which is parametrized by $d\mapsto (d, d\tilde{d}_1, d \tilde{d}_2, d\tilde{d}_3)$.
Along the ball centered at a point on the null geodesic in the direction of $(\tilde{d}_1, \tilde{d}_2, \tilde{d}_3)$, we have \[\begin{split} r&=\sqrt{d^2+2sd Z+s^2},\\ u&=d-  \sqrt{d^2+2sd Z+s^2},\\ \frac{y^1}{r}&= \frac{ d\tilde{d}_1+s\sin\hat{\theta} \sin \hat{\phi} }{ \sqrt{d^2+2sd Z+s^2} }, etc. \end{split}\]
where  \[Z=\tilde{d}_1 \sin\hat{\theta} \sin \hat{\phi}+  \tilde{d}_2\sin\hat{\theta} \cos \hat{\phi}+\tilde{d}_3\cos\hat{\theta}. \]

The pull-back of the global coordinate $(u, r, \theta, \phi)$  under $\mathfrak F_{d}$ defines functions on $B^3$ depending on $d$. As $d\rightarrow \infty$ we have  
\begin{equation}\label{limit_coordinate_bondi}
\lim_{d\rightarrow \infty} \mathfrak F_{d}^* u= -sZ, \lim_{d\rightarrow \infty}\mathfrak F_{d}^* \theta=\tilde{\theta}, \lim_{d\rightarrow\infty} \mathfrak F_{d}^*\phi=\tilde{\phi}, \end{equation} where $\tilde{\theta}, \tilde{\phi}$ are defined such that $\tilde d_1=\sin\tilde{\theta}\sin\tilde{\phi}, \tilde d_2=\sin\tilde{\theta} \cos\tilde\phi$, and  $\tilde{d}_3=\cos\tilde\theta$.

\section{Unit sphere at null infinity in Bondi-Sachs coordinates}
The spacetime metric in Bondi-Sachs coordinates is given by
\begin{align*}
- \lt( 1 - \frac{M}{r} +O(r^{-2}) \rt) du^2 - 2 \lt( 1 + O(r^{-2}) \rt) dudr
- 2 \lt( U_A^{(-2)} + O(r^{-1}) \rt) dudv^A  \\
 + ( r^2 \tilde\sigma_{AB} + r C_{AB} + O(1)) dv^A dv^B. 
\end{align*}

Substituting $u = t - r$, the metric becomes, up to lower order terms,
\begin{align*}
- \lt( 1 - \frac{M}{r} \rt) dt^2 +  \lt(  1 + \frac{M}{r} \rt) dr^2 -\frac{2M}{r} dtdr 
- 2 U_A^{(-2)} (dt -dr)dv^A + ( r^2 \tilde\sigma_{AB} + r C_{AB} ) dv^A dv^B.
\end{align*}

The unit timelike normal of $t=d$ slice is given by
\begin{align*}
\vec{n} = \lt( 1 + \frac{M}{r} \rt) \pl_t + \frac{M}{r} \pl_r + \frac{U_A^{(-2)}}{r} \frac{\pl_A}{r} + O(r^{-2}).
\end{align*}
We compute
\begin{align*}
\langle \na_{\pl_r} \pl_r, \pl_t \rangle &= -\frac{1}{2}\frac{M_u}{r} + O(r^{-2}),\\
\langle \na_{\pl_A} \pl_B, \pl_t \rangle &= - \frac{r}{2} (C_{AB})_u + O(1),
\end{align*}
to get the second fundamental form of $t=d$ slice \begin{align}\label{slice 2nd ff}
\begin{split} k_{rr} &= \frac{1}{2} \frac{M_u}{r} + O(r^{-2})\\
k_{AB} &= \frac{r}{2}(C_{AB})_u + O(1).
\end{split} 
\end{align}

A null geodesic with $u=0, \theta=\tilde \theta, \phi=\tilde \phi$ corresponds to points with the new coordinates \[(t, y^1, y^2, y^3)=(d, d\tilde d_1, d\tilde d_2, d\tilde d_3).\] Let $d_i=d\tilde d_i$. We consider the sphere $\Sigma_d$ of (Euclidean) radius $1$ centered at a point $ (d, d_1, d_2, d_3) $ on the null geodesic and the ball $B_d$ bounded by $\Sigma_d$ in $t$-slice.  Namely,
\begin{align}\label{unit_sphere}
\Sigma_{d}&= \{(t, y^1, y^2, y^3)| \,\,t=d, \sum_i (y^i-d_i)^2 =1\},\\
\Sigma_d(s)&= \{(t, y^1, y^2, y^3)| \,\,t=d, \sum_i (y^i-d_i)^2 =s^2\},\\
B_d &= \{(t, y^1, y^2, y^3)| \,\,t=d, \sum_i (y^i-d_i)^2 \le 1\}.
 \end{align}

In this article, we study the Wang-Yau quasi-local mass of the family of surfaces $\Sigma_d$ defined in \eqref{unit_sphere} as $d\rightarrow \infty$ using the frame work outlined in Section 2. Namely, we consider a family of embedding 
\[
\mathfrak F_{d}=(s,\hat{\theta}, \hat{\phi})\rightarrow (t, y^1, y^2, y^3)=(d, d\tilde{d}_1+s\sin\hat{\theta} \sin \hat{\phi}, d \tilde{d}_2+ s\sin\hat{\theta} \cos \hat{\phi}, d\tilde{d}_3+s\cos\hat{\theta}).
\] 
In particular, $\mathfrak F_{d}$ maps the sphere of radius $s$, $\Sigma(s)$ in $B^3$ onto $\Sigma_d (s)$. The pull-backs of $M$, $U_A^{(-2)} $ and $C_{AB}$ under $\mathfrak F_{d}$ defines tensors on $B^3$ depending on $d$. By \eqref{limit_coordinate_bondi}, their limits as $d \rw \infty$ depend only on $sZ$. We define the following:
\begin{defn} 
We define $F(x)$, $P_{AB}(x)$ and $Q_A(x)$ to be functions of a single variable $x$ such that
\begin{align*}
F(sZ) &  =\lim_{d \rw \infty} M \\
P_{AB}(sZ) &= \lim_{d\rw\infty} C_{AB}\\
Q_A(s Z) &= \lim_{d\rw\infty} U_A^{(-2)}.
\end{align*}
\end{defn}
We use $F'$, $P'_{AB}$ and $Q'_A$ to denote the derivative of these functions with respect to $x$.

We consider the following two functions
$
(\cos\tilde\theta \cos\tilde\phi) \sin\hat\theta \cos\hat \phi + (\cos\tilde\theta \sin\tilde\phi) \sin\hat\theta \sin\hat\phi - \sin\tilde\theta \cos\hat\theta$ and $
 -\sin\tilde\phi \sin\hat\theta \cos\hat\phi + \cos\tilde\phi \sin\hat\theta \sin\hat\phi $. Together with $Z= \sin\tilde{\theta}\sin\tilde{\phi}\sin\hat{\theta} \sin \hat{\phi}+  \sin\tilde{\theta} \cos\tilde\phi\sin\hat{\theta} \cos \hat{\phi}+\cos\tilde\theta\cos\hat{\theta}$, they form an orthogonal basis of first eigenfunctions on $S^2$. We refer to these two functions as $Z^A$. 

In terms of $Z$ and $Z^A$, the transformation formula \cite[page 3]{Vaidya} gives
\begin{align}\label{drdvA}
\begin{split}
dr &=  Z ds + s Z_b du^b + O(d^{-1}) \\
dv^A &= (\frac{1}{r} Z^A) ds + (\frac{s}{r} Z^A_b) du^b + O(d^{-2}).
\end{split}
\end{align}
Let $\bar g$ be the pull-back of the metric on the hypersurface $t=d$  by $\mathfrak F_d$.  
In terms of the coordinate system $\{ s, u^a \}$ on $B^3$, we have
\begin{equation} \label{new_metric}
\begin{split}
\bar g_{ss}= & 1+ \frac{1}{d} \lt( F(sZ) Z^2 + 2Q_A(sZ) Z Z^A + P_{AB}(sZ) Z^A Z^B \right) + O(\frac{1}{d^2}) \\
\bar g_{sa}=& \frac{s}{d} \left( F(sZ)Z Z_a + Q_A(sZ) (Z Z^A_a + Z_a Z^A) +  P_{AB}(sZ) Z^A_a Z^B  \right) + O(\frac{1}{d^2}) \\
\bar g_{{a}{b}}= &s^2 \tilde \sigma_{ab} +  \frac{s^2}{d} \left(  F(sZ) Z_a Z_b + Q_A(sZ) (Z_a Z^A_b + Z_b Z^A_a) + P_{AB}(sZ) Z^A_a Z^B_b \right)  + O(\frac{1}{d^2}) .
\end{split}
\end{equation}
We first compute geometric data on $\Sigma_d(s)$.
\begin{lem} On $\Sigma_d(s)$,
\begin{align*}
\sigma_{ab}^{(-1)} &= s^2 \Big[ F(sZ) Z_a Z_b  + Q_A(sZ)  (Z_a Z^A_b + Z_b Z^A_a ) + P_{AB}(sZ)  Z^A_a Z^B_b \Big]\\
\frac{1}{2} \lt[ \tilde\na^a \tilde\na^b \sigma^{(-1)}_{ab} - tr \sigma^{(-1)} - \tilde\Delta (tr \sigma^{(-1)})\rt] &= s^2 \Big[ -\frac{1}{2}sF'(sZ) Z(1-Z^2)-F(sZ) (1-2Z^2)\\
&\qquad + \lt( sQ'_A(sZ) Z^2 - sQ_A'(sZ)  + 4Q_A (sZ) Z \rt)Z^A \\
&\qquad + \lt( s^2P_{AB}'' (sZ) + sP'_{AB}(sZ) Z + 4P_{AB}(sZ) \rt)Z^A Z^B \Big]
\end{align*}
\end{lem}
\begin{remark}
In the proof, we denote functions such as $F(sZ) $, $F'(sZ) $ and $Q_A(sZ)$ by $F$, $F'$ and $Q_A$.
\end{remark}
\begin{proof}
On $\Sigma_d$, we have
\begin{align*}
\tilde\na^a \tilde\na^b \sigma^{(-1)}_{ab} 
&= F'' (1-Z^2)^2  - 7 F' Z(1-Z^2)^2 -3F(1-3Z^2) \\
&\quad + \lt[ -2Q_A''Z (1-Z^2) -6Q_A'(1-Z^2) + 8Q_A'Z^2 + 18Q_A Z \rt]Z^A \\
&\quad + \lt[ P_{AB}'' Z^2 + 7 P_{AB}'Z + 9P_{AB} \rt]Z^A Z^B\\
\tilde\Delta (tr \sigma^{(-1)}) &= F'' (1-Z^2)^2 - 6 F' Z(1-Z^2)^2+ F(6Z^2-2) \\
&\quad -2 \lt[ Q_A''Z(1-Z^2) -6Q_A'Z^2 + 2Q_A' - 6Q_A Z \rt]Z^A \\
&\quad - \lt[ P_{AB}''(1-Z^2) - 6P_{AB}'Z -6P_{AB} \rt]Z^A Z^B
\end{align*}
The computation on $\Sigma_d(s)$ is similar. We get a factor of $s$ after each derivative.
\end{proof}
\begin{lem} On $\Sigma_d$,
\begin{align*}
(\alpha_H^{(-1)})_a &= - F'ZZ_a + \frac{1}{4} F''(1-Z^2) Z_a + \frac{1}{4} P''_{AB} Z_a Z^A Z^B + \frac{1}{2} P'_{AB} Z^A_a Z^B \\
2 \tilde\na^a (\alpha_H^{(-1)})_a &= \frac{1}{2} F'''(1-Z^2)^2 - 4F''Z(1-Z^2) - 2F'(1-3Z^2)\\
&\quad + \lt( \frac{1}{2} P'''_{AB}(1-Z^2) - 4 P''_{AB} Z - 6 P'_{AB} \right) Z^A Z^B.
\end{align*}
\end{lem}
\begin{proof}
The unit normal of $\Sigma_d$ is $\nu = \pl_s + O(d^{-1})$. By (\ref{drdvA}), we have
\begin{align*}
\pl_s &= Z \pl_r + \frac{1}{d} Z^A \pl_A + O(d^{-2}),\\
\pl_a &= Z_a \pl_r + \frac{1}{d} Z^A_a \pl_A  + O(d^{-2}).
\end{align*}
By (\ref{slice 2nd ff}), we get
\begin{align*}
-k(\nu, \pl_a) &= \frac{1}{2} \frac{M_u}{d} Z_a Z - \frac{1}{2d} (C_{AB})_u Z^A_a Z^B + O(d^{-2}),\\
\tr_\Sigma k &= -\frac{1}{2} \frac{M_u}{d} (1-Z^2) - \frac{1}{2d} (C_{AB})_u Z^A Z^B +O(d^{-2}).
\end{align*}
The assertion follows from $\alpha_H = - k(\nu,\pl_a) + \pl_a \frac{\tr_\Sigma k}{|H|} + O(d^{-2})$.

\end{proof}
\section{The expansion of the Wang-Yau quasi-local mass }
We consider the Wang-Yau quasi-local mass on the unit sphere constructed in the previous section.
\begin{theorem}\label{energy}
For $T_0 = (1,0,0,0)$, 
\begin{align}
\begin{split}
E (\Sigma_d,X,T_0) 
&= \frac{1}{8\pi d^2} \Bigg[ \int_{B^3} \frac{1}{8} \tilde\sigma^{AD} \tilde\sigma^{BE} (C_{AB})_u (C_{DE})_u - \det (h_0^{(-1)} - h^{(-1)}) \\
&\quad\qquad + \frac{1}{4} \int_{S^2} (tr_\Sigma k^{(-1)})^2 - \tau^{(-1)} \tilde\Delta(\tilde\Delta + 2) \tau^{(-1)} \Bigg] + O(d^{-3})
\end{split}
\end{align}
where $\tau^{(-1)}$ is the solution to the optimal embedding equation
\begin{align*}
 \tilde \Delta (\tilde \Delta + 2 )\tau^{(-1)} &= \frac{1}{2} F'''(1-Z^2)^2 - 4F''Z(1-Z^2) - 2F'(1-3Z^2)\\ &\quad +  \lt( \frac{1}{2} P'''_{AB}(1-Z^2) - 4 P''_{AB} Z - 6 P'_{AB} \right) Z^A Z^B.
\end{align*} 
\end{theorem}
\begin{proof}
We write
\[
 E (\Sigma_d,X,T_0) = E_{BY}(\Sigma_d) + (E_{LY}(\Sigma_d)- E_{BY}(\Sigma_d)) + ( E (\Sigma_d,X,T_0) - E_{LY})
\]
where $E_{BY} $ and $E_{LY}$ denote the Brown-York mass and the Liu-Yau mass, respectively. From Lemma 3.1 of \cite{Chen-Wang-Wang-Yau}, we conclude
\[
E_{BY} =\frac{1}{8\pi d^2}  \int_{B^3} \frac{|k^{(-1)}|^2 - (tr k^{(-1)})^2}{2} - \det(h_0^{(-1)} - h^{(-1)}) + O(d^{-3}),
\]
where we also use the vacuum constraint equation 
\[
R= |k|^2 - (tr k)^2.
\]
It is easy to see that 
\[
E_{LY}- E_{BY} =  \frac{1}{32\pi d^2} \int_{S^2} (tr_\Sigma k^{(-1)})^2 + O(d^{-3}).
\]
From the second variation of the Wang-Yau mass in \cite{Chen-Wang-Yau1,Chen-Wang-Yau2}, we have
\[
  E (\Sigma_d,X,T_0) - E_{LY} =\frac{1}{32\pi d^2} \int_{S^2} \tau^{(-1)} \tilde\Delta(\tilde\Delta + 2) \tau^{(-1)} + O(d^{-3}).
 \]
 Finally, we apply \eqref{slice 2nd ff} to evaluate $|k^{(-1)}|$ and $tr k^{(-1)}$.
\end{proof}

\section{Evaluating the qausi-local mass}
Recall the $O(\frac{1}{d})$ terms of the metric coefficients on $B_d$
\begin{align*}
\bar g_{ss}^{(-1)} &= F(sZ) Z^2 + 2Q_A(sZ) Z Z^A + P_{AB}(sZ) Z^A Z^B \\
\bar g_{as}^{(-1)} &= s \lt[ F(sZ)Z Z_a + Q_A(sZ) (Z Z^A_a + Z_a Z^A) + P_{AB}(sZ)Z^A_a Z^B \rt]\\
\bar g_{ab}^{(-1)} &= s^2 \lt[ F(sZ) Z_a Z_b + Q_A(sZ) (Z_a Z^A_b + Z_b Z^A_a) + P_{AB}(sZ) Z^A_a Z^B_b \rt]
\end{align*}

To apply Theorem \ref{energy}, we need to compute $h_0^{(-1)} - h^{(-1)}$ and $\tau^{(-1)}$. We first derive a formula for $h_0^{(-1)} - h^{(-1)}$. 
\begin{lem}\label{h_0-h} 
Let $\mathcal{A}_{AB}(Z,s)$ be a trace-free, symmetric 2-tensor that solves the ODE \begin{align}\label{lin iso emb ode}
\mathcal{A}''_{AB}(Z,s)(1-Z^2) - 6\mathcal{A}'_{AB}(Z,s)Z - 4\mathcal{A}_{AB}(Z,s) = -\frac{s^3}{2} P''_{AB}(sZ) - \frac{s^2}{2}P'_{AB}(sZ)Z - 2s P_{AB}(sZ),
\end{align}
for each $0 < s \le 1$. Here $\mathcal{A}'_{AB}$ means $\frac{\pl \mathcal{A}_{AB}}{\pl Z}$. Then the difference of second fundamental forms on the sphere of radius $s$ is given by
\begin{align*}
  & h_0^{(-1)} - h^{(-1)} \\
=& - \mathcal{A}''_{AB} Z_a Z_b Z^A Z^B  + (\frac{s^2}{2} P'_{AB}(sZ) - 2 \mathcal{A}'_{AB}) \lt( Z_a Z^A_b + Z_b Z^A_a \rt) Z^B \\
&+ \lt(  \mathcal{A}'_{AB}Z   + \mathcal{A}_{AB} - \frac{s}{2}P_{AB}(sZ) \rt) Z^A Z^B \tilde\sigma_{ab} + \lt( sP_{AB}(sZ) - \frac{s^2}{2} P'_{AB}(sZ) Z - 2\mathcal{A}_{AB} \rt)Z^A_a Z^B_b.
\end{align*}
\end{lem}
\begin{proof}
We start with $h^{(-1)}$. The unit normal is given by 
\[ \bar\nu = \lt( 1 - \frac{\bar g_{ss}^{(-1)}}{2d}\rt) \lt( \pl_s - \frac{s^{-2} \tilde\sigma^{ab} \bar g_{as}^{(-1)} }{d} \pl_b \rt) + O(d^{-2}). \]
We compute
\[ h_{ab} = \frac{1}{2} ( \langle D_{\pl_a} \bar\nu, \pl_b \rangle + \langle D_{\pl_b} \bar\nu, \pl_a \rangle ) = s \tilde \sigma_{ab} + \frac{1}{d} \lt( \frac{1}{2}\pl_s \bar g_{ab}^{(-1)} - \frac{\tilde\na_a \bar g_{bs}^{(-1)} + \tilde \na_b \bar g_{as}^{(-1)}}{2} - \frac{\bar g_{ss}^{(-1)}}{2} s\tilde\sigma_{ab} \rt)+ O(d^{-2}). \]

For $h_0^{(-1)}$, we expand the isometric embedding $X$ as
\[
X=s\tilde X + \frac{1}{d}X^{(-1)} + O(d^{-2})
\]
where $\tilde X$ denote the unit sphere in $\R^3$. We decompose $X^{(-1)}$ into $X^{(-1)} = \alpha^a \pl_a + \beta \nu$. The linearized isometric embedding equation reads
\begin{equation}\label{isometric_first_order} \sigma^{(-1)}_{ab} = s^2( \tilde\sigma_{ac}\tilde\na_b \alpha^c + \tilde\sigma_{bc}\tilde\na_a \alpha^c) + 2\beta s \tilde \sigma_{ab}.  \end{equation}
From the computation in \cite[pages 938-939]{Wang-Yau2}, \eqref{isometric_first_order}  implies that
\begin{equation}\label{second_ff_first}   h_0^{(-1)} = -\tilde\na_a \tilde\na_b \beta -\beta \tilde \sigma_{ab} + \frac{1}{s} \sigma^{(-1)}_{ab}.  \end{equation}
Putting these together, we obtain \begin{align} \label{second_ff_difference}
h_0^{(-1)} - h^{(-1)} = -\tilde\na_a\tilde \na_b \beta - \beta\tilde\sigma_{ab} + \frac{1}{s}\sigma^{(-1)}_{ab} - \frac{1}{2}(\pl_s \bar g_{ab})^{(-1)} + \frac{\tilde\na_a \bar g_{bs}^{(-1)} + \tilde\na_b \bar g_{as}^{(-1)}}{2} + \frac{\bar g_{ss}^{(-1)}}{2} s \tilde \sigma_{ab}.
\end{align}
To solve $\beta$, we consider the expansion of the Gauss curvature $K(d,s)$ of $\Sigma_d(s)$. Let
\[
K(d,s) =     \frac{1}{s^2} + \frac{1}{d} K^{(-1)} + O(d^{-2})
\]
On the one hand, from the metric expansion, we get 
\[
 K^{(-1)} =\frac{1}{s^2} \lt (  -\tilde\na^a \tilde\na^b \sigma^{(-1)}_{ab} + tr_{S^2}\sigma^{(-1)} + \tilde\Delta  tr_{S^2} \sigma^{(-1)} \rt).
\] 
On the other hand, combining \eqref{second_ff_first}   and the Gauss equation, we conclude that 
\[
 K^{(-1)} =\frac{2}{s}  (\tilde\Delta + 2)\beta
\]
As a result, $\beta$ is the solution of
\begin{align}\label{lin_iso_emb}
2 s (\tilde\Delta + 2)\beta = -\tilde\na^a \tilde\na^b \sigma^{(-1)}_{ab} + tr_{S^2}\sigma^{(-1)} + \tilde\Delta  tr_{S^2} \sigma^{(-1)}.
\end{align}
For the right hand side, we compute 
\begin{align*}
-\tilde\na^a \tilde\na^b \sigma^{(-1)}_{ab} + tr_{S^2}\sigma^{(-1)} + \tilde\Delta  tr_{S^2} \sigma^{(-1)}
&= s^3 F'(sZ) Z (1-Z^2) + s^2 F(2-4Z^2) \\
&\quad + s^3 Q'_A(sZ)(2-2Z^2) Z^A - 8 s^2 Q_A(sZ) ZZ^A\\
&\quad + \lt( -s^4 P''_{AB}(sZ) - s^3 P'_{AB}(sZ) - 4s^2 P_{AB}(sZ) \rt) Z^A Z^B
\end{align*}
On the other hand, 
let $\mathcal{F}$ and  $\mathcal{Q}_A$ be an antiderivative of $F$ and  $Q_A$ respectively, and $\mathcal{A}_{AB}$ satisfy (\ref{lin iso emb ode}). One verifies that  
\begin{equation}\label{beta_formula} \beta = \frac{\mathcal{F}(sZ)}{2} Z + \mathcal{Q}_A(sZ)Z^A + \mathcal{A}_{AB}(Z,s) Z^AZ^B \end{equation}
solves the linearized isometric embedding equation \eqref{lin_iso_emb} since, for a trace-free, symmetric 2-tensor $\mathcal{A}_{AB}(Z,s)$, 
\begin{align*}
(\tilde\Delta + 2) \lt( \mathcal{A}_{AB}(Z,s) Z^A Z^B \rt)= \lt( \mathcal{A}''_{AB}(Z,s) (1-Z^2) -6 \mathcal{A}'_{AB}(Z,s) Z - 4 \mathcal{A}_{AB}(Z,s) \rt) Z^A Z^B.
\end{align*} 
We are ready compute \eqref{second_ff_difference} where $\beta$ is given in \eqref{beta_formula}. We have
\begin{align*}
- \tilde\na_a \tilde\na_b \beta - \beta \tilde\sigma_{ab} &= - \frac{s^2}{2} F' Z Z_a Z_b + \frac{s}{2} F Z^2 \tilde \sigma_{ab} \\&\quad - s^2 Q_A'Z_aZ_b Z^A + s Q_A ZZ^A \tilde\sigma_{ab}\\
&\quad - \mathcal{A}''_{AB} Z_a Z_b Z^AZ^B -2 \mathcal{A}'_{AB} (Z_a Z^A_b + Z_b Z^A_a ) Z^B \\
&\quad + \lt( sP_{AB} -2 \mathcal{A}_{AB} \rt) Z^A_a Z^B_b + (\mathcal{A}'_{AB} Z + \mathcal{A}_{AB}) Z^A Z^B \tilde\sigma_{ab} - \frac{1}{s} \sigma^{(-1)}_{ab} ,
\end{align*}

\begin{align*}
\frac{1}{s}\sigma^{(-1)}_{ab} - \frac{1}{2} \pl_s \bar g_{ab}^{(-1)} &= -\frac{s^2}{2} \lt( F'Z Z_a Z_b + Q_A' Z(Z_a Z^A_b + Z_b Z^A_a) + P'_{AB} Z Z^A_a Z^B_b \rt) 
\\
\frac{1}{2}(\tilde\na_a \bar g_{bs}^{(-1)} + \tilde\na_b \bar g_{as}^{(-1)}) &=  s^2 F'Z Z_a Z_b - sFZ^2 \tilde\sigma_{ab} \\
&\quad + \frac{s^2}{2} Q_A' Z (Z_a Z^A_b + Z_b Z^A_a) + s^2 Q'_A Z_a Z_b Z^A - 2s Q_A Z Z^A \tilde \sigma_{ab}\\
&\quad + \frac{s^2}{2} P'_{AB} (Z_a Z^A_b + Z_b Z^A_a) Z^B - s P_{AB} Z^A Z^B \tilde\sigma_{ab} + \frac{1}{s} \sigma_{ab}^{(-1)}\\
\frac{1}{2} \bar g_{ss}^{(-1)} s\tilde\sigma_{ab} &= s \lt( \frac{1}{2}F Z^2 + Q_A Z Z^A + \frac{1}{2} P_{AB} Z^A Z^B \rt) \tilde \sigma_{ab}.
\end{align*}

We see that terms involving $F,Q_A$ cancel and the result has the asserted form.
\end{proof}
Next we compute $\tau^{(-1)}$. 
\begin{lemma}\label{lemma4.2}
Define the second order differential operator 
 \[ L\mathcal{G}(Z) = \lt[ (1-Z^2) \mathcal{G}'\rt]'(Z) - 4 \mathcal{G}'(Z)Z - 6 \mathcal{G}(Z). \]
Let $\mathcal{B}_{AB}(Z)$ be a traceless, symmetric 2-tensor that solves the ODE \begin{align}\label{ODE_tau}
L(L+2)\mathcal{B}_{AB} = \frac{1}{2}P'''_{AB}(Z)(1-Z^2) - 4 P''_{AB}(Z)Z - 6P'_{AB}(Z).
\end{align}
Then \[
\tau^{(-1)} = Z \mathcal{F}(Z) + \mathcal{B}(Z)_{AB}Z^A Z^B
\]
solves the leading order of optimal embedding equation 
\begin{align*}
 \tilde \Delta (\tilde \Delta + 2 )\tau^{(-1)} &= \frac{1}{2} F'''(1-Z^2)^2 - 4F''Z(1-Z^2) - 2F'(1-3Z^2)\\ &\quad +  \lt( \frac{1}{2} P'''_{AB}(Z)(1-Z^2) - 4 P''_{AB}(Z)Z - 6 P'_{AB}(Z) \rt) Z^A Z^B.
\end{align*}
\end{lemma}
\begin{proof}
The equation is linear. We look for $\tau^{(-1)}_1$ and $\tau^{(-1)}_2 $ such that 
\[
\begin{split}
 \tilde \Delta (\tilde \Delta + 2 )\tau^{(-1)}_1 &= \frac{1}{2} F'''(1-Z^2)^2 - 4F''Z(1-Z^2) - 2F'(1-3Z^2),\\
 \tilde \Delta (\tilde \Delta + 2 )\tau^{(-1)}_2 &= \lt( \frac{1}{2}P'''_{AB}(Z)(1-Z^2) - 4 P''_{AB}(Z)Z - 6 P'_{AB}(Z) \rt) Z^AZ^B.
\end{split}
\]
From Lemma 3.3 of \cite{Vaidya}, $\tau^{(-1)}_1 = Z \mathcal{F}(Z)$ solves the first equation
\[
 \tilde \Delta (\tilde \Delta + 2 ) (Z \mathcal{F}(Z))=\frac{1}{2} F'''(1-Z^2)^2 - 4F''Z(1-Z^2) - 2F'(1-3Z^2).
\]
It is straightforward to verify that $\tau^{(-1)}_2  = \mathcal{B}_{AB}(Z) Z^A Z^B$
solves the second equation if the traceless, symmetric 2-tensor $\mathcal{B}_{AB}(Z)$ solves \eqref{ODE_tau}.\end{proof}
We are ready to state the main theorem for the quasi-local mass,
\begin{theorem}\label{main_theorem}
For $T_0 = (1,0,0,0)$ and $X$ solves the leading order term of the optimal embedding equation, the Wang-Yau quasi-local energy
\begin{align*}
\begin{split}
  & E(\Sigma_d,T_0,X) \\
=& \frac{1}{d^2} \Bigg[ \int_{B^3} \frac{1}{8} \sum_{A,B} P'_{AB}(sZ) P'_{AB}(sZ) - \det(h_0^{(-1)} - h^{(-1)}) \\
&+  \frac{1}{4}\int_{S^2}  \frac{1}{4}(P'_{AB} Z^A Z^B)^2 - \mathcal{B}_{DE} Z^D Z^E  \lt(\frac{1}{2}P'''_{AB}(Z)(1-Z^2) - 4 P''_{AB}(Z)Z - 6 P'_{AB}(Z) \rt) Z^AZ^B \Bigg] + O(d^{-3})
\end{split}
\end{align*} 
where $h_0^{(-1)} - h^{(-1)}$ is as determined in Lemma 5.1 and $\mathcal{B}_{AB}$ is as determined in Lemma 5.2.
\end{theorem}
\begin{proof}
We start with Theorem 4.1 in which $h_0^{(-1)} - h^{(-1)}$ is as determined in Lemma 5.1 and $\tau^{(-1)}$ is as determined in Lemma 5.2. We simplify the expression
\[
\begin{split}
\int_{S^2} (tr_\Sigma k^{(-1)})^2 - \tau^{(-1)} \tilde\Delta(\tilde\Delta + 2) \tau^{(-1)} = & \int_{S^2} \frac{1}{4} F^2(1-Z^2)^2 -\tau_1^{(-1)} \tilde\Delta(\tilde\Delta + 2) \tau_1^{(-1)}\\
& + \int_{S^2} \frac{1}{4} \lt( P'_{AB}(Z) Z^A Z^B \rt)^2 -\tau_2^{(-1)} \tilde\Delta(\tilde\Delta + 2) \tau_2^{(-1)}.
\end{split}
\]
We have 
\[
\int_{S^2}  \frac{1}{4} F^2(1-Z^2)^2 -\tau_1^{(-1)} \tilde\Delta(\tilde\Delta + 2) \tau_1^{(-1)}= 0
\]
by \cite[(3.6)]{Vaidya}. This finishes the proof of the theorem.
\end{proof}

In particular, we observe that the answer depends on the leading order term of the news function on $B^3$ since both ODEs in Lemma 5.1 and Lemma 5.2 are linear ODEs where the right-hand side depends on $P_{AB}$ and their derivatives. In general, we do not have explicit solutions to these ODEs. In the following section, we compute the quasi-local mass explicitly for a few special examples. 

\section{Special cases}
Write $E(\Sigma_d,T_0,X) = d^{-2} E^{(-2)} + O(d^{-3})$. We evaluate $E^{(-2)}$ for a few special cases of $P_{AB}$. Let $p_{AB}, q_{AB}$ be two constant symmetric traceless 2-tensors. 
\begin{prop}
If $P_{AB}(x) = p_{AB} + q_{AB} x,$ $E^{(-2)}=0$.
\end{prop}\begin{proof}
One verifies that \begin{align*} 
\mathcal{A}_{AB}(Z,s) &= \frac{s}{2} p_{AB} + \frac{s^2Z}{4} q_{AB} \\
\mathcal{B}_{AB}(Z) &= - \frac{1}{4} q_{AB}
\end{align*} solve \eqref{lin iso emb ode} and \eqref{ODE_tau} respectively.  Direct computation shows that $h_0^{(-1)} - h^{(-1)}=0$. Hence,
\[ E^{(-2)} =\frac{1}{8 \pi} \lt ( \frac{1}{8 } \sum_{A,B} q_{AB}q_{AB} \cdot \frac{4\pi}{3} + \frac{1}{4} \int_{S^2} \frac{1}{4} \lt( q_{AB} Z^A Z^B \rt)^2 + \frac{1}{4}  q_{DE} Z^DZ^E \cdot (-6 q_{AB} Z^A Z^B)  \rt) =0, \] 
where we used the identity \begin{align}\label{integral}
\int_{S^2} Z^A Z^B Z^D Z^E = \frac{4\pi}{15} (\delta^{AB} \delta^{DE} + \delta^{AD}\delta^{BE} + \delta^{AE}\delta^{BD}).
\end{align}
\end{proof}  
\begin{prop}
If $P_{AB}(x) = p_{AB}x^2.$ Then $E^{(-2)} = \frac{1}{20} \sum_{A,B} p_{AB}p_{AB}$.
\end{prop}
\begin{proof}
One verifies that
\begin{align*}
\mathcal{A}_{AB}(Z,s) &= s^3 \lt( \frac{(Z)^2}{6} + \frac{1}{3} \rt)p_{AB}\\
\mathcal{B}_{AB}(Z) &= -\frac{Z}{6}p_{AB}
\end{align*}
solve \eqref{lin iso emb ode} and \eqref{ODE_tau} respectively. Direct computation shows that \begin{align*}
h_0^{(-1)} - h^{(-1)} = \frac{s^3}{3} \Big( Z^A Z^B \tilde\sigma_{ab} - Z_aZ_b Z^AZ^B + Z \lt( Z_a Z^A_b + Z_b Z^A_a \rt) Z^B - \lt( (Z)^2+2 \rt) Z^A_aZ^B_b \Big)p_{AB}.
\end{align*}
We compute
\begin{align*}
\lt| h_0^{(-1)} - h^{(-1)} \rt|^2_{\tilde\sigma} &= \frac{s^2}{9} \Big( 9\lt( p_{AB} Z^A Z^B \rt)^2 + (2(Z)^2-8) \delta^{AD} Z^B Z^E p_{AB} p_{DE} \\
&\qquad\quad + ((Z)^2+2)^2 \delta^{AD}\delta^{BE} p_{AB} p_{DE} \Big),\\
\mbox{tr}_{\tilde\sigma} \lt( h_0^{(-1)} - h^{(-1)} \rt) &= sZ^A Z^B p_{AB}
\end{align*}
to get
\begin{align*}
\det \lt( h_0^{(-1)} - h^{(-1)} \rt) &= \frac{1}{2} \lt( \mbox{tr}_{\tilde\sigma} \lt( h_0^{(-1)} - h^{(-1)} \rt) - \lt| h_0^{(-1)} - h^{(-1)} \rt|^2_{\tilde\sigma}\rt)\\
&= -\frac{s^2}{18} \Big( \lt( 2(Z)^2-8 \rt) \delta^{AD} Z^B Z^E + \lt( (Z)^2+2 \rt)^2 \delta^{AD}\delta^{BE}  \Big)p_{AB} p_{DE}.
\end{align*}
Denote $|p|^2 = \sum_{A,B} p_{AB}p_{AB}$. The volume integral contributes 
\begin{align*}
\frac{1}{3} \int_{S^2} \lt[ \frac{(Z)^2}{2} |p|^2  + \frac{1}{18} \lt( (2(Z)^2-8) \delta^{AD} Z^B Z^E p_{AB} p_{DE} + ((Z)^2+2)^2 |p|^2 \rt) \rt] = \frac{4\pi}{9}|p|^2 
\end{align*}
and the surface integral contributes
\begin{align*}
\frac{1}{4} \int_{S^2} (Z)^2 (p_{AB} Z^A Z^B)^2 - \frac{10}{3} (Z)^2 Z^D Z^E p_{DE} Z^A Z^B p_{AB} =- \frac{2\pi}{45}|p|^2, \end{align*}
where we used the identity $\int_{S^2} (Z)^2 Z^A Z^B Z^D Z^E = \frac{4\pi}{105}(\delta^{AB} \delta^{DE} + \delta^{AD}\delta^{BE} + \delta^{AE}\delta^{BD})$.
\end{proof}

\end{document}